\documentclass[11pt]{article}

\usepackage{times,graphics,psfrag,amsthm,amsmath}

\newtheorem{theorem}{Theorem}
\newtheorem{lemma}{Lemma}

\begin{document}

\title{\bf Clique in 3-track interval graphs is APX-hard\thanks{%
Supported in part by NSF grant DBI-0743670.}}

\author{Minghui Jiang\medskip\\
Department of Computer Science,
Utah State University,
Logan, UT 84322, USA\smallskip\\
\texttt{mjiang@cc.usu.edu}}

\maketitle

\begin{abstract}
Butman, Hermelin, Lewenstein, and Rawitz proved that
\textsc{Clique} in $t$-interval graphs is NP-hard for $t \ge 3$.
We strengthen this result to show that
\textsc{Clique} in $3$-track interval graphs is APX-hard.
\end{abstract}



\section{Introduction}

We prove the following theorem:

\begin{theorem}\label{T1}
\textsc{Clique} in $3$-track interval graphs is APX-hard.
\end{theorem}

\paragraph{Preliminaries}

Given a set $X = \{ x_1, \ldots, x_n \}$ of $n$ boolean variables
and a set $C = \{c_1, \ldots, c_m \}$ of $m$ clauses,
where each variable occurs at most (resp.~exactly) $p$ times in the clauses,
and each clause is the conjunction (resp.~disjunction) of exactly $q$ literals,
\textsc{$p$-Occ-Max-E$q$-CSAT}
(resp.~\textsc{E$p$-Occ-Max-E$q$-SAT})
is the problem of finding an assignment for $X$ that satisfies the maximum
number of clauses in $C$.

\begin{lemma}
\textsc{$12$-Occ-Max-E$2$-CSAT} is APX-hard.
\end{lemma}

\begin{proof}
It is known that \textsc{E$3$-Occ-Max-E$2$-SAT} is APX-hard~\cite{BK03}.
For each disjunctive clause $x_1 \lor x_2$,
we can construct a set of $6$ conjunctive clauses
$$
x_1 \land y
\qquad
x_1 \land \bar y
\qquad
x_2 \land y
\qquad
x_2 \land \bar y
\qquad
x_1 \land \bar x_2
\qquad
\bar x_1 \land x_2
$$
where $y$ is an additional dummy variable.
If both $x_1$ and $x_2$ are false, then
none of the $6$ clauses is satisfied.
If either $x_1$ or $x_2$ is true, then
exactly $2$ of the $6$ clauses are satisfied.
Thus we have a gap-preserving L-reduction~\cite{PY91}
from
\textsc{E$3$-Occ-Max-E$2$-SAT}
to
\textsc{$12$-Occ-Max-E$2$-CSAT}
with $\alpha=2$ and $\beta=1/2$.
\end{proof}

\section{Proof of Theorem~\ref{T1}}

We prove that \textsc{Clique} in $3$-track interval graphs is APX-hard
by an L-reduction from
\textsc{$12$-Occ-Max-E$2$-CSAT}.
Given an instance $(X, C)$ of
\textsc{$12$-Occ-Max-E$2$-CSAT},
we construct a $3$-track interval graph $G$ as the intersection graph
of a set of $24n + m$ $3$-track intervals:
\begin{itemize}

\item
$12$ copies of a $3$-track interval for the positive literal $x_i$
of each variable $x_i \in X$;

\item
$12$ copies of a $3$-track interval for the negative literal $\bar x_i$
of each variable $x_i \in X$;

\item
$1$ copy of a $3$-track interval for each clause $c_k \in C$.

\end{itemize}
Each $3$-track interval in our construction
is the union of three open intervals,
one interval on each track,
of integer endpoints between $-(n+1)$ and $n+1$.

For each variable $x_i$,
the $3$-track interval for the positive literal $x_i$
is the union of the following three intervals
$$
\text{track~1: } (-i, i)
\qquad
\text{track~2: } (i, n+1)
\qquad
\text{track~3: } (i, i+1)
$$
and the $3$-track interval for the negative literal $\bar x_i$
is the union of the following three intervals
$$
\text{track~1: } (i, n+1)
\qquad
\text{track~2: } (-i, i)
\qquad
\text{track~3: } (-(i+1), -i)
$$

Assume without loss of generality that no clause contains both the positive
literal and the negative literal of the same variable.
For each clause $c_k$,
we construct one $3$-track interval following one of four cases:
\begin{enumerate}

\item
$c_k = x_i \land x_j$, $i \le j$.
$$
\text{track~1: } (-(n+1), -j)
\qquad
\text{track~2: } (-(n+1), i)
\qquad
\text{track~3: } \left\{\begin{array}{ll}
(i+1, j) &\text{if } j > i+1
\\
(-1, 1) &\text{if } j = i \text{ or } i+1
\end{array}\right.
$$

\item
$c_k = \bar x_i \land \bar x_j$, $i \le j$.
$$
\text{track~1: } (-(n+1), i)
\qquad
\text{track~2: } (-(n+1), -j)
\qquad
\text{track~3: } \left\{\begin{array}{ll}
(-j, -(i+1)) &\text{if } j > i+1
\\
(-1, 1) &\text{if } j = i \text{ or } i+1
\end{array}\right.
$$

\item
$c_k = x_i \land \bar x_j$, $i < j$.
$$
\text{track~1: } (i, j)
\qquad
\text{track~2: } (-(n+1), -j)
\qquad
\text{track~3: } (-1, i)
$$

\item
$c_k = \bar x_i \land x_j$, $i < j$.
$$
\text{track~1: } (-(n+1), -j)
\qquad
\text{track~2: } (i, j)
\qquad
\text{track~3: } (-i, 1)
$$

\end{enumerate}

\begin{figure}[htbp]
\psfrag{x1}{$x_1$}
\psfrag{x2}{$x_2$}
\psfrag{x3}{$x_3$}
\psfrag{x4}{$x_4$}
\psfrag{x1'}{$\bar x_1$}
\psfrag{x2'}{$\bar x_2$}
\psfrag{x3'}{$\bar x_3$}
\psfrag{x4'}{$\bar x_4$}
\psfrag{c1}{$c_1$}
\psfrag{c2}{$c_2$}
\psfrag{c3}{$c_3$}
\psfrag{c4}{$c_4$}
\psfrag{-5}{$-5$}
\psfrag{0}{$0$}
\psfrag{5}{$5$}
\psfrag{track 1}{track 1}
\psfrag{track 2}{track 2}
\psfrag{track 3}{track 3}
\centering
\resizebox{\linewidth}{!}{\includegraphics{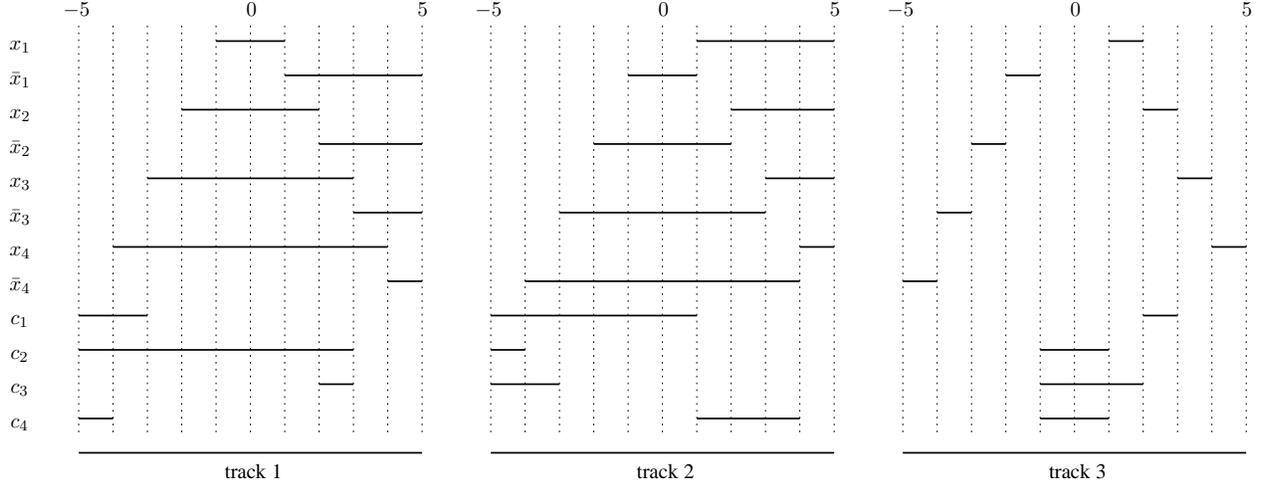}}
\caption{\small%
The set of $3$-track intervals for a \textsc{$12$-Occ-Max-E$2$-CSAT} instance
of $n=4$ variables and $m=4$ clauses
$c_1 = x_1 \land x_3$,
$c_2 = \bar x_3 \land \bar x_4$,
$c_3 = x_2 \land \bar x_3$,
and
$c_4 = \bar x_1 \land x_4$.
Duplicate $3$-track intervals for each literal are omitted from the figure.}
\label{fig:example}
\end{figure}

This completes the construction.
We give an example in Figure~\ref{fig:example}.
The reduction clearly runs in polynomial time.
We have the following lemma:

\begin{lemma}\label{L1}
There is an assignment for $X$ that satisfies at least $z$ clauses in $C$
if and only if
$G$ has a clique of size at least $w = 12n + z$.
\end{lemma}

\begin{proof}
The following observations can be easily verified:
\begin{itemize}

\item
For any two variables $x_i$ and $x_j$, $i \neq j$,
the $3$-track intervals for the literals of $x_i$
overlap with
the $3$-track intervals for the literals of $x_j$.

\item
For any two clauses $c_k$ and $c_l$, $k \neq l$,
the $3$-track interval for $c_k$
overlaps with
the $3$-track interval for $c_l$.

\item
For each variable $x_i$,
the $3$-track intervals for the positive literal of $x_i$
are disjoint from
the $3$-track intervals for the negative literal of $x_i$.

\item
For each clause $c_k$,
the $3$-track interval for $c_k$
is disjoint from
the $3$-track intervals for the two literals in $c_k$,
and overlaps with the $3$-track intervals for the other literals.

\end{itemize}

We first prove the direction implication.
Suppose
there is an assignment for $X$ that satisfies at least $z$ clauses in $C$.
We select a subset of pairwise-intersecting $3$-track intervals as follows.
For each clause $c_k$ in $C$,
select the corresponding $3$-track interval if the clause is satisfied.
Then, for each variable $x_i$ in $X$,
select the $12$ copies of the $3$-track interval
for the negative literal of $x_i$ if the variable is true,
and select the $12$ copies $3$-track interval
for the positive literal of $x_i$ if the variable is false.
Thus we obtain a clique of size at least $w = 12 n + z$ in $G$.

We next prove the reverse implication.
Suppose $G$ has a clique of size at least $w = 12 n + z$.
Note that in our construction
the number of $3$-track intervals for each literal
is at least the number of $3$-track intervals
for all clauses that contain the literal.
Thus, by replacing vertices,
any clique can be converted into a clique of at least the same size
in \emph{canonical} form, which includes, for each variable $x_i$ in $X$,
either all $12$ copies of the $3$-track interval
for the positive literal of $x_i$,
or all $12$ copies of the $3$-track interval
for the negative literal of $x_i$.
Assign $x_i$ false if the clique includes the $3$-track intervals for its
positive literal,
and assign $x_i$ true if the clique includes the $3$-track intervals for its
negative literal.
Thus we obtain an assignment for $X$ that satisfies at least $z$ clauses in $C$.
\end{proof}

Let $z^*$ be the maximum number of clauses in $C$ that can be satisfied by
an assignment of $X$, and let $w^*$ be the maximum size of a clique in $G$.
By Lemma~\ref{L1}, we have $w^* = 12n + z^*$.
Since each clause in $C$ is the conjunction of exactly two literals
of the variables in $X$,
we have $n \le 2m$.
Moreover, since a random assignment for $X$ satisfies
each clause in $C$ with probability at least $1/4$,
we have $z^* \ge m/4 \ge n/8$.
It follows that
$$
w^* = 12n + z^* \le (8\cdot 12 + 1) z^* = 97 z^*.
$$

Consider any clique of size $w$ in $G$.
Following the reverse implication in the proof of Lemma~\ref{L1},
we can find an assignment for $X$
that satisfies at least $z = w - 12 n$ clauses in $C$.
Note that
$$
|z^* - z| \le |w^* - w|.
$$
Thus we have an L-reduction
from
\textsc{$12$-Occ-Max-E$2$-CSAT}
to
\textsc{Clique} in $3$-track interval graphs
with $\alpha = 97$ and $\beta = 1$.
This completes the proof of Theorem~\ref{T1}.

\paragraph{Postscript}
This note was written in April 2010. The author would like to thank
St\'ephane Vialette for bringing the open questions of
Butman et~al.~\cite{BHLR07} to his attention in 2010,
and for notifying him of the recent results of Francis et~al.~\cite{FGO12},
who proved the APX-hardness of \textsc{Clique} in several classes of
multiple-interval graphs (including $3$-track interval graphs)
using the subdivision technique of~\cite{CC07}.

\end{document}